\documentclass[10pt]{article}

\usepackage{amsmath}
\usepackage{amssymb}  
\usepackage[T1]{fontenc}
\usepackage{pifont}
\usepackage{pstricks}
\usepackage{amsfonts}
\usepackage{graphicx}
\usepackage{amsmath}
\usepackage{authblk}
\usepackage{pstricks} 
\usepackage{pstricks-add}
\usepackage{caption}
\usepackage{placeins}
\usepackage{pstricks}
\usepackage{pdfpages}
\usepackage{framed}
\usepackage{bigints}
\usepackage{cancel}
\usepackage{bm}
\usepackage{epstopdf}
\usepackage{tikz}
\usepackage{relsize}

\def\rr2dot{\mathop{\bf r}\limits}
\def\x2dot{\mathop{x}\limits}
\def\y2dot{\mathop{y}\limits}

\def\bfy2dot{\mathop{\bf y}\limits}
\def\z2dot{\mathop{z}\limits}

\def\csi2dot{\mathop{\xi}\limits}
\def\et2dot{\mathop{\eta}\limits}
\def\bet2dot{\mathop{\beta}\limits}
\def\t2dot{\mathop{\theta}\limits}
\def\s2dot{\mathop{\sigma}\limits}
\def\d2dot{\mathop{\delta }\limits}
\def\q2dot{\mathop{q}\limits}
\def\l2dot{\mathop{\lambda}\limits}
\def\ps2dot{\mathop{{\cal E}}\limits}
\def\tet2dot{\mathop{\theta}\limits}
\def\bfx2dot{\mathop{\bf x}\limits}
\def\bfy2dot{\mathop{\bf y}\limits}
\def\bfq2dot{\mathop{\bf q}\limits}
\def\bfr2dot{\mathop{\bf r}\limits}
\def\bbfq2dot{\mathop{\bar {\bf q}}\limits}
\def\w2{\mathop{W}\limits}
\def\xgrande2dot{\mathop{\bf X}\limits}
\def\p02dot{\mathop{P}\limits}
\def\a2dot{\mathop{A}\limits}

\newtheorem{teo}{Theorem}
\newtheorem{lem}{Lemma}
\newtheorem{prop}{Proposition}

\newtheorem{rem}{Remark}
\newtheorem{exe}{Example}

\title{On the commutation of variation and differentiation in nonholonomic Systems: A Chetaev-based approach}

\author{F.~Talamucci}
\affil{{\it DIMAI, Dipartimento di Matematica e Informatica ``Ulisse Dini''},\\
{\it	Universit\`a degli Studi di Firenze, Italy}\\
{\it	e-mail: federico.talamucci@unifi.it}}
\date{}

\begin{document}
		
	\bibliographystyle{plain}
	
	\setcounter{equation}{0}

	\maketitle
	
	\vspace{.5truecm}
	
	\noindent
	{\bf 2020 Mathematics Subject Classification:} 70F25, 70H25, 70H03.
	
	\vspace{.5truecm}
	
	\noindent
	{\bf Keywords:} 
	
Nonholonomic systems, Chetaev condition, Lagrangian derivative, Variational principles, Transposition rules, Dynamic compensation, Frobenius integrability.
	
	\vspace{.5truecm}
	
\begin{abstract}
	
The derivation of the equations of motion for nonholonomic systems remains a central issue in analytical mechanics, primarily due to the tension between the d’Alembert-Lagrange differential principle and integral variational approaches. This study investigates the validity of the commutation relation between the variational operator and the time derivative, which is a geometric identity in holonomic manifolds but becomes problematic when dealing with velocity-dependent constraints.By analyzing the "transposition rule," we define a formal relationship between the Chetaev variation and the total variation of the constraints. We show that the simultaneous requirement of kinematically admissible variations and the fulfillment of the Chetaev condition is generally incompatible with the standard commutation rule, unless a specific geometric condition—encoded through a skew-symmetric algebraic structure and the Lagrangian derivative of the constraints—is satisfied. Furthermore, this work extends the analysis to systems with multiple constraints introducing the concept of dynamic compensation. While Frobenius' Theorem provides a "static" criterion for integrability based on individual vector fields, our results suggest that dynamic consistency according to Chetaev’s principle emerges as a collective phenomenon. We demonstrate that even when individual constraints are intrinsically non-integrable, their interactions can "cancel out" deviations from holonomy, maintaining global consistency. Notably, we show that for systems with high constraints 
this property is satisfied regardless of the constraints' form. These findings broaden the class of analyzable physical systems, suggesting that dynamic consistency is a resilient property that persists even in the absence of simple geometric integrability.
	
\end{abstract}

\section{Introduction}

This study stems from a long-standing issue regarding the dynamics of constrained systems, namely the interchangeability of the operators $\delta$ and $d/dt$, where $\delta$ represents the variational operator (virtual displacement) and $d/dt$ the total derivative with respect to time. More precisely, we refer to the transposition rule as the identity
\begin{equation}
	\label{commrel}
	\delta \frac{dq_i}{dt} = \frac{d}{dt} \delta q_i
\end{equation}
which asserts the commutation relation between the variation and the time derivative for any of the Lagrangian coordinates $q_i$, $i=1, \dots, n$. This identity implies that the variation of the velocity is equal to the rate of change of the virtual displacement. In a broader sense, it states that the processes of varying a path and evolving along a path are commutative.

In the context of holonomic mechanics, where constraints depend strictly on positions and time, the commutation relation (\ref{commrel}) is an established and uncontested property. J.~L.~Lagrange \cite{lagrange} implicitly relied on the identity $\delta (dq) = d(\delta q)$ as a necessary step to derive the equations of motion from the Principle of Virtual Work and d'Alembert's Principle.

For holonomic systems, this identity is a consistent result of the independence between virtual variations and time evolution; its formal justification was later refined. V.~Volterra \cite{volterra} and other mathematicians provided a rigorous foundation within the calculus of variations, confirming that for admissible virtual displacements in holonomic manifolds, the operators must commute to maintain geometric consistency.

While (\ref{commrel}) is regarded as an undisputed geometric identity in the study of holonomic systems, its validity and application become contentious when dealing with nonholonomic systems—those characterized by non-integrable constraints on velocities. The transition to nonholonomic systems (e.g., a sphere rolling without slipping or a wheeled robot) led to a bifurcation in mathematical modeling, centered on how constraints interact with the variational process.

Significant attention was drawn to this subject by H.~Hertz \cite{hertz}, who suggested that for nonholonomic systems, the variation and the derivative do not necessarily commute. In 1896, O.~H\"older \cite{holder} demonstrated that if constraints are imposed on the kinematic velocities \textit{before} varying the trajectory, the identity $\delta \dot{q} = \frac{d}{dt} \delta q$ fails. This framework was modernized by V.~V.~Kozlov \cite{kozlov} as ``vakonomic dynamics'' (Variational Axiomatic Nonholonomic). In this model, constraints are incorporated into the Lagrangian via multipliers before variation, leading to equations where the multipliers $\lambda_\nu$ often appear under a time derivative ($\dot{\lambda}_\nu$). While mathematically elegant and useful in optimal control theory or biological growth models, this approach generally fails to describe the physical reality of classical mechanical systems.

Important treatises on mechanics have expressed specific positions regarding commutation. The exchangeability of the variation operator $\delta$ and the time derivative $d/dt$ is a foundational yet subtle point in analytical mechanics. Its validity is closely tied to the definition of virtual displacements and the nature of the system's constraints.
In his influential work \cite{pars} L.~A.~Pars adopts a definitive stance, asserting that the commutativity of $\delta$ and $d/dt$ is an inherent property of the variational calculus. Pars emphasizes that if the variation $\delta$ is understood as a displacement between two points at the same instant $t$ (synchronous variation), the identity:$$\delta \left( \frac{dq}{dt} \right) = \frac{d}{dt} (\delta q)$$is an analytical necessity rather than a physical assumption. Crucially, Pars maintains the validity of this rule even for nonholonomic systems. He argues that the complexities of nonholonomic dynamics do not arise from a failure of the commutation rule itself, but rather from the constraints imposed on the variations $\delta q_i$. In Pars' view, the operators commute by definition of the variational process, and any discrepancy in the equations of motion must be attributed to the admissibility of the displacements, not to the operator algebra.

Similarly, in his treatise \cite{lurie}, A.~I.~Lurie treats the commutativity of $\delta$ and $d/dt$ as a rigorous consequence of synchronous variation. Under the condition that the time variable remains unvaried ($\delta t = 0$), the variation $\delta$ is defined as a differential shift between the actual trajectory and a kinematically admissible comparison trajectory at the same instant $t$. Lurie demonstrates that this identity holds because $\delta$ is treated as an operator acting on the configuration space independently of the temporal evolution parameter. For Lurie, this commutativity is a structural necessity for the standard derivation of the Euler-Lagrange equations via Hamilton’s Principle.

A more nuanced and critical perspective is offered by Mei Fengxiang \cite{mei}. He highlights that the "commutativity condition" is not a universal law but a choice of variational procedure that can lead to physical discrepancies in complex systems. While Mei aligns with the classical view for holonomic systems, he explores the tension between "commutative" and "non-commutative" variations in nonholonomic contexts. He notes that if one forces the variations to satisfy the constraints (as in the Vakonomic approach), the standard commutation rule may lead to equations of motion that differ from those derived via the d'Alembert-Lagrange principle. Mei’s work emphasizes that in systems with non-integrable constraints, the operator $[\delta, d/dt]$ may not vanish if the definition of the virtual displacement deviates from the "frozen time" convention.

\subsubsection*{The Chetaev condition}

To resolve the discrepancies between vakonomic predictions and physical reality, N.~G.~Chetaev \cite{chetaev}, \cite{chetaev62} proposed a rigorous foundation based on the d’Alembert-Lagrange principle. Chetaev argued that virtual displacements $\delta q$ are not arbitrary but must be ``compatible'' with the constraints. For a nonholonomic constraint $g_\nu(q, \dot{q}, t) = 0$, the variations must satisfy the Chetaev condition:
\begin{equation}
	\label{dcdefint}
	\sum_{i=1}^n \frac{\partial g_\nu}{\partial \dot{q}_i} \delta q_i = 0
\end{equation}
This assumption ensures that the virtual work of the constraint forces vanishes, maintaining consistency with the behavior of physical systems like wheels or skates. The Chetaev condition, which lacks a purely theoretical or mathematical justification, should be regarded as a physical choice that defines the ideal nature of the constraint. The widespread adoption of this condition in modeling is also due to the analytical convenience it offers in handling various aspects, such as the derivation of the equations of motion.

Starting from d’Alembert’s Principle (which is a differential, and thus instantaneous, principle):
\begin{equation}
	\sum_{i=1}^n \left[ \frac{d}{dt} \left( \frac{\partial \mathcal{L}}{\partial \dot{q}_i} \right) - \frac{\partial \mathcal{L}}{\partial q_i} \right] \delta q_i = 0
\end{equation}
In this framework, the $\delta q_i$ represent instantaneous virtual displacements. Since there is no need to vary an entire trajectory over a time interval, the question of whether $\delta$ and $d/dt$ commute along a path does not even arise. It is sufficient to require that the $\delta q_i$ satisfy the Chetaev condition (\ref{dcdefint}). In this case, the Lagrange equations with multipliers emerge directly as an algebraic consequence of the linear dependence of the variations, without ever invoking the commutation relation (\ref{commrel}):
\begin{equation}
	\frac{d}{dt} \frac{\partial L}{\partial \dot{q}_i} - \frac{\partial L}{\partial q_i} = \sum_{\nu=1}^\kappa \lambda_\nu \frac{\partial g_\nu}{\partial \dot{q}_i}, \qquad i=1, \dots, n
\end{equation}
where $\kappa < n$ is the number of constraints.

\subsubsection*{The variational approach}

Why, then, do we speak of commutation? The discussion regarding commutation (\ref{commrel}) becomes necessary only if one intends to derive the equations of motion from Hamilton's Principle (the integral approach); that is, by varying the action $S = \int L \, dt$, where the entire path, rather than just an instantaneous state, is subject to variation.

When applying a variational approach to nonholonomic systems, constraints can be incorporated directly into the action functional by considering the integral $\int (L + \sum_{\nu=1}^\kappa \lambda_\nu g_\nu) dt$ \cite{arnoldkozlovnei}, \cite{ray}, \cite{kozlov}. This method fundamentally redefines the role of the Lagrange multipliers. During the variation of the functional, the dependence of the constraints $g_\nu$ on the velocities $\dot{q}$ becomes a critical factor. Specifically, the total variation of the constraint functions is defined as:
\begin{equation}
	\label{dvdefint}
	\sum_{i=1}^n \frac{\partial g_\nu}{\partial q_i} \delta q_i + \sum_{i=1}^n \frac{\partial g_\nu}{\partial \dot{q}_i} \delta \dot{q}_i = 0
\end{equation}
where the vanishing of this quantity ensures that the varied curves consistently satisfy the constraints; that is, the varied curve is still kinematically admissible. In the variational process, integrating the term containing $\delta \dot{q}_i$ (or $\frac{d}{dt}\delta q_i$) by parts forces the multipliers $\lambda_\nu$ to appear under a time derivative ($\dot{\lambda}_\nu$). This profoundly alters the structure of the resulting equations of motion:
\begin{equation}
	\label{eqvak}
	\frac{d}{dt} \left( \frac{\partial L}{\partial \dot{q}_i} + \sum_{\nu=1}^\kappa \lambda_\nu \frac{\partial g_\nu}{\partial \dot{q}_i} \right) - \left( \frac{\partial L}{\partial q_i} + \sum_{\nu=1}^\kappa \lambda_\nu \frac{\partial g_\nu}{\partial q_i} \right) = 0, \qquad i=1, \dots, n
\end{equation}
leading to a higher-order system where the multipliers $\lambda_\nu$ act as additional dynamical variables. While mathematically elegant, this approach does not always describe the observed motion of standard nonholonomic mechanical systems; consequently, the physical validity of these "vakonomic" equations (\ref{eqvak}) has been widely questioned \cite{cronstrom}, \cite{flanneryenigma}, \cite{lemos}, \cite{liromp}.

In this framework, the commutation relation (\ref{commrel}) holds by mathematical necessity, as it is the very hypothesis that allows the elimination of $\delta \dot{q}$ terms by shifting them onto $\delta q$ via integration by parts. The vakonomic approach thus forces the system to behave as if a commutative variation could always be defined, often yielding results at odds with empirical reality.

A broader generalization of the variational framework allows for the violation of the commutation property between differentiation and variation. This line of reasoning traces back to Suslov, who introduced specific correction terms to account for the non-commutative nature of virtual displacements in nonholonomic mechanics. In this context, the principle of stationary action is assumed to take the form:
\begin{equation}
	\int_{t_1}^{t_2} \left[ \delta L + \sum_{i=1}^n \frac{\partial L}{\partial \dot{q}_i} \left( \frac{d}{dt}\delta q_i - \delta \dot{q}_i \right) \right] dt = 0.
\end{equation}
The specific form of the resulting dynamics depends on the assumption made regarding the relationship between $\delta \dot{q}_i$ and $\frac{d}{dt}\delta q_i$. A comprehensive reference for these issues, offering both a historical overview and a rigorous treatment of the pseudo-coordinate technique, is Neimark and Fufaev \cite{neimark}. Recent studies, such as Llibre \cite{llibre}, further investigate the variational origin of nonholonomic equations by employing "transpositional relations" of the form $\delta \dot{q}_i - \frac{d}{dt}\delta q_i = \sum_{j=1}^n W_{i,j}\delta q_j$, where the coefficients $W_{i,j}$ must be specified based on the system's geometry.
A significant and noteworthy point of reference for the application of variational principles in the context of nonholonomic mechanics is the work \cite{krup} of Krupkov\'a.

\subsection{Transposition rule}

Let us rewrite (\ref{dcdefint}) and (\ref{dvdefint}) using the following symbols:

\begin{itemize}
	\item The Chetaev variation:
	\begin{equation}
		\label{dcdef}
		\delta^{(c)}g_\nu := \sum_{i=1}^n \frac{\partial g_\nu}{\partial \dot{q}_i}\delta q_i
	\end{equation}
	
	\item The total variation:
	\begin{equation}
		\label{dvdef}
		\delta^{(v)} g_\nu := \sum_{i=1}^n \frac{\partial g_\nu}{\partial q_i}\delta q_i + \sum_{i=1}^n \frac{\partial g_\nu}{\partial \dot{q}_i}\delta \dot{q}_i
	\end{equation}
\end{itemize}

The starting point of the analysis is the following relationship between these variations:
\begin{equation}
	\label{transpruleg}
	\delta^{(v)} g_\nu - \frac{d}{dt} \left( \delta^{(c)}g_\nu \right) = \sum_{i=1}^n \frac{\partial g_\nu}{\partial \dot{q}_i} \left( \delta \dot{q}_i - \frac{d}{dt}\delta q_i \right) - \sum_{i=1}^n \mathcal{D}_i g_\nu \delta q_i
\end{equation}
where the operator $\mathcal{D}_i$ is the Lagrangian derivative:
\begin{equation}
	\label{derlagr}
	\mathcal{D}_i g_\nu = \frac{d}{dt} \left( \frac{\partial g_\nu}{\partial \dot{q}_i} \right) - \frac{\partial g_\nu}{\partial q_i}
\end{equation}
We refer to relation (\ref{transpruleg}), the verification of which is immediate, as the transposition rule. It shows that the difference between the total variation and the rate of change of the Chetaev variation depends on two factors:
\begin{itemize}
	\item The potential non-commutation of the coordinates (the first term on the right);
	\item The geometry of the constraint itself, captured by the operator $\mathcal{D}_i$ (the second term on the right).
\end{itemize}

Depending on the assumptions made, relation (\ref{transpruleg}) provides meaningful information regarding the state of the system and restrictions on the quantities involved. The following situations can be identified:
\begin{itemize}
	\item If the commutation relation (\ref{commrel}) holds, then the assumption $\delta^{(c)}g_\nu = 0$ (Chetaev condition) must necessarily be supplemented by:
	\begin{equation}
		\label{hv}
		\delta^{(v)} g_\nu = -\sum_{i=1}^n \mathcal{D}_i g_\nu \delta q_i
	\end{equation}
	This imparts a specific structure to the variations of the paths: the varied curve violates the constraints unless $\sum_{i=1}^n \mathcal{D}_i g_\nu \delta q_i = 0$.
	
	\item Combining $\delta^{(v)}g_\nu = 0$ with commutation (\ref{commrel})—fundamental to vakonomic theory—leads to the following regarding the validity of the Chetaev condition:
	\begin{equation}
		\label{vv}
		\frac{d}{dt} \left( \delta^{(c)}g_\nu \right) = \sum_{i=1}^n \mathcal{D}_i g_\nu \delta q_i
	\end{equation}
	which holds only if $\sum_{i=1}^n \mathcal{D}_i g_\nu \delta q_i = 0$.

	\item The simultaneous validity of both assumptions $\delta^{(c)}g_\nu = 0$ and $\delta^{(v)}g_\nu = 0$ cannot imply commutation (\ref{commrel}) unless:
	\begin{equation}
		\label{transpruledvdc0}
		\sum_{i=1}^n \mathcal{D}_i g_\nu \delta q_i = 0, \qquad \nu = 1, \dots, \kappa
	\end{equation}
	The above equality is therefore a necessary (but not sufficient) condition for commutation in the case indicated.
\end{itemize}

\begin{rem}
The combinations of hypotheses are widely discussed in the literature regarding nonholonomic models under several designations: (\ref{hv}) is referred to as H\"older vakonomic variation, (\ref{vv}) as vakonomic variation, and the simultaneous validity of $\delta^{(c)}g_\nu = 0$ and $\delta^{(v)}g_\nu = 0$ (not necessarily requiring commutation) pertains to the Suslov variation \cite{guo}.
\end{rem}

\section{Problem statement}

Among the various possibilities for investigation, we found it of interest to further explore the following question: which systems of constraints $g_1, \dots, g_\kappa$ satisfy property \eqref{derlagr} within the space of displacements allowed by the Chetaev condition?

The main motivation is the recurrence of this condition across various implications. Furthermore, the condition is satisfied by integrable constraints, as we shall discuss later; we therefore expect that, within the framework of nonholonomic constraints, this condition allows us to capture a class that displays similarities with holonomic ones.

The physical context for this study is a mechanical system defined by $n$ generalized coordinates, subject to $\kappa$ general constraints. These constraints are functions of the coordinates $\mathbf{q}=(q_1, \dots, q_n)$, the velocities $\dot{\mathbf{q}}$, and time $t$:
\begin{equation}
	\label{vincg}
	g_\nu (\mathbf{q}, \dot{\mathbf{q}}, t) = 0, \qquad \nu=1, \dots, \kappa < n.
\end{equation}
To ensure that these constraints are physically meaningful and not redundant, we impose the rank condition:
\begin{equation}
	\label{gind}
	\mathrm{rank} \left( \frac{\partial (g_1, \dots, g_\kappa)}{\partial (\dot{q}_1, \dots, \dot{q}_n)} \right) = \kappa.
\end{equation}
This condition states that the Jacobian matrix of the constraints with respect to the generalized velocities must have full rank $\kappa$, ensuring the functional independence of the constraints.

We address the problem of establishing necessary and sufficient conditions under which the set of constraints $\{g_\nu\}$, $\nu=1, \dots, \kappa$, satisfies $\sum_{i=1}^n \mathcal{D}_i g_\nu \delta q_i = 0$ for all virtual displacements satisfying the constraints. More specifically, the task is to determine the sets of constraints such that:
\begin{equation}
	\label{derlagrsomma0}
	\sum_{i=1}^n \mathcal{D}_i g_\nu \delta q_i = 0 \quad \text{for each } \nu=1, \dots, \kappa \text{ and for all } \delta \mathbf{q} \text{ s.t. } \delta^{(c)}g_\mu = \sum_{i=1}^n \frac{\partial g_\mu}{\partial \dot{q}_i}\delta q_i = 0, \mu=1, \dots, \kappa.
\end{equation}
As previously mentioned, the problem can be framed as a search for necessary conditions for commutation \eqref{commrel}, for virtual displacements that satisfy the Chetaev condition \eqref{dcdefint} and are admissible for the constraint variation \eqref{dvdefint}.

\subsection{Integrability and Integrating Factors}

We identify two specific classes of constraint functions that provide solutions to problem \eqref{derlagrsomma0}, based on their integrability properties:
\begin{itemize}
	\item \textbf{Integrable constraints}: The constraint can be integrated as
	\begin{equation}
		\label{gint}
		g_\nu = \frac{d}{dt} f_\nu(\mathbf{q},t);
	\end{equation}
	\item \textbf{Constraints with an integrating factor}: The constraint admits a factor $\phi_\nu$ such that
	\begin{equation}
		\label{gif}
		\phi_\nu (\mathbf{q},t) g_\nu = \frac{d}{dt} f_\nu (\mathbf{q},t).
	\end{equation}
\end{itemize}

In the first case \eqref{gint}, it is a well-known property that $\mathcal{D}_i (d f_\nu / dt) = 0$ for any $i=1, \dots, n$, and thus \eqref{derlagrsomma0} is identically verified. Concerning the second case, it can be easily checked that a constraint of type \eqref{gif} satisfies:
\begin{equation}
	\label{derlagrgif}
	\phi_\nu \mathcal{D}_i g_\nu = - \frac{d \phi_\nu}{dt} \frac{\partial g_\nu}{\partial \dot{q}_i}
\end{equation}
from which it follows that
\begin{equation}
	\phi_\nu \sum_{i=1}^n \mathcal{D}_i g_\nu \delta q_i = - \frac{d \phi_\nu}{dt} \sum_{i=1}^n \frac{\partial g_\nu}{\partial \dot{q}_i} \delta q_i = -\frac{d \phi_\nu}{dt} \delta^{(c)}g_\nu = 0
\end{equation}
since the displacements satisfy the Chetaev condition. Therefore, equality \eqref{derlagrsomma0} holds wherever $\phi_\nu \neq 0$.

The distinction between these categories is fundamental: while holonomic constraints satisfy the property by their very nature, nonholonomic constraints requiring an integrating factor only align with this behavior when virtual displacements are restricted to those allowed by Chetaev's postulate. In other terms, the integrating factor $\phi_\nu$ introduces a rescaling that preserves the geometry of the constraint manifold but modifies its Lagrangian derivative.

\subsection{An Effective Characterization of the Lagrangian Derivative}

\noindent
The following characterization, rooted in the independence (\ref{gind}) of the constraints, provides a significant analytical advantage for the subsequent derivations.
\begin{prop}
For any fixed index $\nu \in \{1, \dots, \kappa\}$, the condition $\sum_{i=1}^n {\cal D}_i g_\nu \delta q_i=0$ holds for all virtual displacements satisfying the constraints if and only if
\begin{equation}
\label{zero2}{\cal D}i g\nu =\sum_{\mu=1}^\kappa\varrho_\mu^{(\nu)}\dfrac{\partial g_\mu}{\partial {\dot q}i}, \qquad i=1, \dots, n
\end{equation}
for some real coefficients $\varrho\mu^{(\nu)}({\bf q}, {\dot {\bf q}},t)$.
\end{prop}

{\bf Proof}. Let $\nu$ be fixed. The condition $\sum_{i=1}^n {\cal D}_i g_\nu \delta q_i=0$ implies that the vector $({\cal D}_i g_\nu)_{i=1}^n$ vanishes on the intersection of the kernels of the linear forms defined by the gradients $(\frac{\partial g_\mu}{\partial \dot{q}_i})_{i=1}^n$. According to the properties of linear orthogonality (or the method of Lagrange multipliers), this occurs if and only if $({\cal D}_i g_\nu)$ is a linear combination of these gradients. Under the assumption of linear independence (\ref{gind}), the coefficients $\varrho_\mu^{(\nu)}$ are uniquely determined at each state $({\bf q}, {\dot {\bf q}},t)$.Conversely, if (\ref{zero2}) holds, then for any displacement $\delta {\bf q}$ satisfying condition (A), it follows that:$$\sum_{i=1}^n {\cal D}_i g_\nu \delta q_i = \sum_{\mu=1}^\kappa \varrho_\mu^{(\nu)} \left( \sum_{i=1}^n \dfrac{\partial g_\mu}{\partial {\dot q}_i} \delta q_i \right) = \sum_{\mu=1}^\kappa \varrho_\mu^{(\nu)} \cdot 0 = 0. \quad \square$$

By virtue of Proposition 1, problem (\ref{derlagrsomma0}) is equivalent to identifying the systems of constraints that satisfy (\ref{zero2}) for each $\nu=1, \dots, \kappa$. This approach effectively removes the virtual displacements $\delta {\bf q}$ from the formulation.

It is worth noting that this condition is satisfied by the two aforementioned categories fulfilling (\ref{derlagrsomma0}):\begin{itemize}\item For exact constraints (\ref{gint}), property (\ref{zero2}) is trivially satisfied with $\varrho_\mu^{(\nu)}=0$ for all $\mu, \nu$.\item For category (\ref{gif}), condition (\ref{zero2}) is satisfied whenever $\Phi_\nu \neq 0$ (recalling (\ref{derlagrgif})) by setting:$$\varrho_\mu^{(\nu)}({\bf q}, {\dot {\bf q}})= -\delta_{\mu,\nu} \frac{1}{\Phi_\nu} \sum_{j=1}^n \frac{\partial \Phi_\nu}{\partial q_j}{\dot q}_j$$where $\delta_{\mu,\nu}$ denotes the Kronecker delta.\end{itemize}

\subsection{Linear homogeneous constraints}

\noindent
It should be observed that applying the operator ${\cal D}_i$ to $g_\nu$ explicitly generates terms involving the accelerations $\ddot{\bf q}$ unless $g_\nu$ is linear in $\dot{\bf q}$. For property (\ref{zero2}) to hold independently of the system's accelerations, the functions $g_\nu$ must be linear with respect to the generalized velocities $\dot{q}_i$. In the non-linear case, (\ref{zero2}) would instead define a set of second-order differential equations restricting the dynamics. This technical consideration motivates focusing the subsequent analysis exclusively on linear homogeneous constraints
\begin{equation}
\label{glinom}g_\nu ({\bf q}, {\dot {\bf q}})=\sum\limits_{j=1}^n a_{\nu,j}({\bf q}){\dot q}_j, \quad \nu=1, \dots, \kappa.
\end{equation}
For constraints of type (\ref{glinom}), the definition (\ref{dcdef}) reduces to
\begin{equation}
	\label{dcdeflin}
	\delta^{(c)}g_\nu=\sum\limits_{j=1}^n a_{\nu, j}({\bf q})\delta q_j.
\end{equation}
Under assumption (\ref{dik}), these equations can be expressed in the explicit form
\begin{equation} 
\label{glinomexpl}
{\dot q}_\nu=\sum\limits_{j=1}^{n-\kappa} b_{\nu, j}({\bf q}){\dot q}_{\kappa+j}, \qquad \nu=1, \dots, \kappa
\end{equation} 

\noindent
where the velocities $({\dot q}_{\kappa+j})_{j=1}^{n-\kappa}$ act as independent generalized velocities, while $({\dot q}_\nu)_{\nu=1}^\kappa$ represent the dependent components determined by the coefficients $b_{\nu,j}(\mathbf{q})$.
The Jacobian matrix (\ref{gind}) is the $\kappa \times n$ matrix with entries $a_{\nu,j}$. 
By applying the operator ${\cal D}_i$ to the linear form (\ref{glinom}), we obtain
\begin{equation}
\label{derlagrlin}{\cal D}_i g_\nu = \sum_{j=1}^n \left( \frac{\partial a_{\nu, i}}{\partial q_j} - \frac{\partial a_{\nu, j}}{\partial q_i} \right) \dot{q}j.
\end{equation}
As anticipated, for linear constraints, the terms involving $\ddot{\bf q}$ vanish identically because $\frac{\partial^2 g_\nu}{\partial {\dot q}_j \partial {\dot q}_i} = 0$. Consequently, relation (\ref{zero2}) becomes a purely kinematic identity:
\begin{equation}
\label{zero2lin}
\sum\limits_{j=1}^n A_{i,j}^{(\nu)}{\dot q}_j = \sum\limits_{\mu=1}^\kappa \varrho_\mu^{(\nu)} a_{\mu,i}, \qquad i=1, \dots, n
\end{equation}
for suitable functions $\varrho_{\mu}^{(\nu)}({\bf q}, {\dot {\bf q}})$, where we define the skew-symmetric tensor
\begin{equation}
\label{ela}A^{(\nu)}{i,j} = \dfrac{\partial a{\nu,i}}{\partial q_j} - \dfrac{\partial a_{\nu,j}}{\partial q_i} = -A_{j,i}^{(\nu)}.
\end{equation}
The identity (\ref{zero2lin}) relates the "curl-like" terms of the constraint coefficients to a linear combination of the constraints themselves. This structure is fundamental for defining the ideality of constraints and for the subsequent derivation of the equations of motion.

Let $A^{(\nu)}$ be the square matrix of order $n$ with elements $A^{(\nu)}_{i,j}$ and let ${\bf a}_\mu$ denote the $\mu$-th row of the constraint matrix $(a_{\nu, j})$. The equalities (\ref{zero2lin}) can be expressed in compact vector-matrix notation as
\begin{equation}
\label{zero2linvm}A^{(\nu)} {\dot {\bf q}} = \sum_{\mu=1}^\kappa \varrho^{(\nu)}_\mu {\bf a}_\mu, \qquad \nu=1, \dots, \kappa
\end{equation}
where ${\dot {\bf q}}$ and ${\bf a}_\mu$ are treated as column vectors.

For the sake of clarity, we restate the Main Problem (\ref{derlagrsomma0}) for the case of linear homogeneous constraints:
\begin{quote}
	Determine the set of linear constraints $\{g_1, \dots, g_\kappa\}$ such that
\begin{equation}
\label{derlagrsomma0lin}
\sum_{i,j=1}^n A^{(\nu)}_{i,j} {\dot q}_j \delta q_i = 0 \quad \forall \nu \in {1, \dots, \kappa}
\end{equation}
for all virtual displacements $\delta {\bf q}$ satisfying $\sum\limits_{i=1}^n a_{\mu, i}({\bf q})\delta q_i = 0$ for $\mu=1, \dots, \kappa$.
\end{quote}
The focus of the following mathematical investigation will be the analysis of property (\ref{zero2linvm}).

\section{The main result}

\noindent
\noindent
Let $D$ denote the determinant of the $\kappa \times \kappa$ submatrix of the constraint coefficients:
\begin{equation}
	\label{det}D = \det \left( a_{\nu,\mu} \right)_{\nu,\mu=1, \dots, \kappa}
\end{equation}
where $a_{\nu,\mu}$ are the coefficients from the constraint equations (\ref{glinom}). By virtue of the non-degeneracy condition (\ref{gind}), we assume without loss of generality that the leading principal minor of order $\kappa$ is non-zero:
\begin{equation}
	\label{dik}D \neq 0.
\end{equation}
Let $D_{r}(\ell)$ be the determinant obtained by replacing the $r$-th column of $D$ with the $\ell$-th column of the full $\kappa \times n$ matrix. Similarly, $D_{r,s}(\ell, m)$ denotes the determinant where the $r$-th and $s$-th columns are replaced by columns $\ell$ and $m$, respectively:
\begin{equation}
	\label{drl}D_r(\ell)= \det \left( {\bf a}_1 \dots \underbrace{{\bf a}_\ell}_{pos., r} \dots {\bf a}_\kappa \right),\quad
	D_{r,s}(\ell,m)=\det \left( {\bf a}_1 \dots \underbrace{{\bf a}_\ell}_{pos., r} \dots \underbrace{{\bf a}_m}_{pos., s} \dots {\bf a}_\kappa \right).
\end{equation}

The main result concerning problem (\ref{derlagrsomma0lin}) is stated in the following

\begin{teo}
Let $\{ g_\nu\}_{\nu=1}^\kappa$ be a set of linear homogeneous constraints (\ref{glinom}) verifying (\ref{dik}). Then, the property (\ref{derlagrsomma0lin}), namely $\sum\limits_{i=1}^\nu {\cal D}_i g_\nu \delta q_i=0$, is satisfied for each $\nu$ and for any $\delta q_i$ fulfilling $\sum\limits_{i=1}^n a_{\mu,i}\delta  q_i=0$ (that is $\delta ^{(c)}g_\mu$) for $\mu=1, \dots, \kappa$ if and only if 
\begin{equation}
	\label{fp}
\mu^{(\nu)}_{i,j}({\bf q})=0 \qquad \textrm{for all}\;\;i,j=1, \dots, n-\kappa\;\;\textrm{and}\;\;\nu=1, \dots, \kappa
\end{equation}
where the functions $\mu^{(\nu)}{i,j}$ are defined as
\begin{align}
	\mu^{(\nu)}_{i,j} &= \sum_{r=1}^\kappa \left( A^{(\nu)}_{r, \kappa+j} D_r(\kappa+i) - A^{(\nu)}_{r, \kappa+i} D_r(\kappa+j) \right) - A^{(\nu)}_{\kappa+i, \kappa+j } D \nonumber \\
	&- \sum_{1 \leq r < s \leq \kappa} A^{(\nu)}_{r,s} D_{r,s}(\kappa+i, \kappa+j) = 0 
	\label{mu}
\end{align}
and columns $\kappa+i$ and $\kappa+j$ are the columns of the matrix $a_{\nu, \ell}$ for 
$\nu=1, \dots, \kappa$ and $\ell=1, \dots, n$.
\end{teo}

\noindent
The proof is provided in the Appendix. We now present some key remarks on the structure of the conditions (\ref{mu}):

\begin{itemize}
	\item Invariance: while the explicit form of (\ref{mu}) depends on the choice of the non-singular minor (\ref{det}), selecting a different minor yields an equivalent set of conditions.
	
	\item Skew-symmetry: the coefficients satisfy $\mu^{(\nu)}_{j,i} = -\mu^{(\nu)}_{i,j}$. This is evident for the terms in the first line of (\ref{mu}), while for the second line, it follows from the property that swapping two columns in a determinant changes its sign. Consequently, for each $\nu$, the matrix $\mathbf{M}^{(\nu)} = (\mu^{(\nu)}_{i,j})$ of order $n-\kappa$ is skew-symmetric, implying $\mu^{(\nu)}_{i,i}=0$. The upper triangular elements $1 \leq i < j \leq n-\kappa$ are sufficient to characterize the system:
\begin{equation}
	\label{matrmu}
	\begin{pmatrix}
		& \mu^{(\nu)}_{1,2} & \cdots & \mu^{(\nu)}_{1,n-\kappa} \\
		& & \ddots & \vdots \\
		& & & \mu^{(\nu)}_{n-\kappa-1,n-\kappa} \\
		& & & 
	\end{pmatrix}
\end{equation}

	\item Number of conditions: as a consequence, the number of independent conditions (\ref{mu}) for each $\nu$ is:\begin{equation}\label{nc}N_c = \frac{(n-\kappa)(n-\kappa-1)}{2}.\end{equation}

	\item Complexity: for a fixed pair $(i,j)$, the number of individual $A \cdot D$ terms in (\ref{mu}) is $\frac{(\kappa+1)(\kappa+2)}{2}$, which is independent of the total dimension $n$.

	\item Constraint Interaction: For $\kappa > 1$, the validity of condition (\ref{fp}) is not a property of individual constraints in isolation (as is the case for integrable constraints). Instead, it arises from the algebraic interaction between them, mediated by the determinants of the coefficients of the entire constraint set.
\end{itemize}

\subsection{Special cases}

\subsubsection{Exact constraints}

\noindent
If the constraint \eqref{glinom} is exact, namely $g_\nu = \frac{d}{dt}f_\nu (\mathbf{q})$, then $a_{\nu,i} = \frac{\partial f_\nu}{\partial q_i}$. Under sufficient regularity assumptions, all coefficients \eqref{ela} vanish identically due to the symmetry of mixed partial derivatives (Schwarz's Theorem):
\begin{equation}
A_{i,j}^{(\nu)} = \frac{\partial^2 f_\nu}{\partial q_j \partial q_i} - \frac{\partial^2 f_\nu}{\partial q_i \partial q_j} = 0, \qquad i,j=1, \dots, n.
\end{equation}
Consequently, the conditions \eqref{mu} are trivially satisfied, consistent with the previous remarks regarding \eqref{gint}.

\subsubsection{Constraints admitting an integrating factor}

\noindent
In the more general case of constraints admitting an integrating factor—for which we have already established the validity of property \eqref{derlagrsomma0}—the conditions \eqref{mu} are likewise recovered. Indeed, a constraint of type \eqref{gif} is necessarily linear; in the autonomous case \eqref{glinom}, it must satisfy
\begin{equation}
\phi_\nu(\mathbf{q}) a_{\nu,j}(\mathbf{q}) = \frac{\partial f_\nu (\mathbf{q})}{\partial q_j}, \quad j=1, \dots, n,
\end{equation}
where the following closure conditions are necessary for the existence of an integrating factor $\phi_\nu$:
\begin{equation}
\label{ccgint}\phi_\nu \left( \frac{\partial a_{\nu,i}}{\partial q_j} - \frac{\partial a_{\nu,j}}{\partial q_i} \right) = a_{\nu, j} \frac{\partial \phi_\nu}{\partial q_i} - a_{\nu, i} \frac{\partial \phi_\nu}{\partial q_j}, \quad i,j=1, \dots, n.
\end{equation}
These conditions are also sufficient in any simply connected open subset by the Poincar\'e Lemma.
To verify \eqref{fp}, we multiply \eqref{mu} by $\phi_\nu$ and substitute \eqref{ccgint}. After some algebraic manipulation, we obtain:
\begin{align*}
\phi_\nu \mu_{i,j}^{(\nu)} &= \sum_{r=1}^\kappa \frac{\partial \phi_\nu}{\partial q_r} \bigg( a_{\nu, \kappa+j}D_r(\kappa+i) - a_{\nu, \kappa+i}D_r(\kappa+j) \&\quad + \sum_{s=r+1}^\kappa a_{\nu, s}D_{r,s}(\kappa+i, \kappa+j) - \sum_{\sigma=1}^{r-1} a_{\nu, \sigma}D_{\sigma, r}(\kappa+i, \kappa+j) \bigg) \\&
\quad + \frac{\partial \phi_\nu}{\partial q_{\kappa+i}} \bigg( \sum_{r=1}^\kappa a_{\nu,r}D_r(\kappa+j) - a_{\nu, \kappa+j}D \bigg) \&\quad - \frac{\partial \phi_\nu}{\partial q_{\kappa+j}} \bigg( \sum_{r=1}^\kappa a_{\nu, r}D_r(\kappa+i) - a_{\nu, \kappa+i}D \bigg).
\end{align*}
By expanding the determinants, one can verify that the terms within each of the three bracketed expressions vanish identically. Specifically:
\begin{itemize}
	\item Vanishing of $\frac{\partial \phi_\nu}{\partial q_{\kappa+i}}$ terms: the expression $\sum_{r=1}^\kappa a_{\nu,r}D_r(\kappa+j) - a_{\nu, \kappa+j}D$ represents the Laplace expansion along a column of a $( \kappa+1 ) \times ( \kappa+1 )$ matrix where the last column is a linear combination of the preceding ones. Since the coefficients $a_{\nu,r}$ and $a_{\nu,\kappa+j}$ are linked by the constraint relations, this determinant possesses linearly dependent columns, thus vanishing identically.
	
	\item Vanishing of the first summation: the terms involving higher-order minors $D_{r,s}$ cancel out due to the antisymmetry property of the determinants. When summing over $r$ and $s$, each minor $D_{r,s}$ appears twice with opposite signs, effectively neutralizing the contribution of the partial derivatives $\frac{\partial \phi_\nu}{\partial q_r}$.
\end{itemize}
 
 This result ensures that the existence of an integrating factor—and thus the existence of an underlying manifold for the constraints—is perfectly captured by the vanishing of the $\mu_{i,j}^{(\nu)}$ coefficients, aligning the geometric and algebraic descriptions of the system.

\subsubsection{The coefficients $\mu_{i,j}^{(\nu)}$ for particular values of $\kappa$}\subsubsection*{The case $\kappa=1$}

\noindent
In this case, only a single constraint $a_{1,1}{\dot q}_1 + a_{1,2}{\dot q}_2 + \dots + a_{1,n}{\dot q}_n = 0$ is imposed. The number of independent conditions \eqref{fp} is $\frac{(n-1)(n-2)}{2}$. Each quantity $\mu_{i,j}$ (for $1 \leq i < j \leq n-1$) in the matrix \eqref{matrmu} consists of three terms and can be expressed as:
\begin{equation}
\label{muk1}\mu_{i,j} = -A_{1, i+1} a_{1, j+1} + A_{1, j+1} a_{1, i+1} - A_{i+1, j+1} a_{1,1}
\end{equation}where we have omitted the superscript $(\nu)=(1)$ for brevity.

\begin{exe}
(magnetic-like coupling) Consider a four-dimensional system ($n=4$) subject to a single nonholonomic constraint ($\kappa=1$). Let ${\bf q} = (x, y, z, w)$ be the generalized coordinates; we define the constraint as
\begin{equation}
-y \dot{x} + x \dot{y} + z \dot{z} + \dot{w} = 0.
\end{equation}
This structure is of particular interest as it represents a magnetic-like coupling. The term $(-y \dot{x} + x \dot{y})$ is formally analogous to the interaction of a particle with a magnetic vector potential, implying that the evolution of $w$ is strictly coupled to the angular momentum in the $(x, y)$ plane.
 
\noindent
 Assuming the condition \eqref{dik} holds (which here requires $a_{1,4} \neq 0$, always satisfied since $a_{1,4}=1$), the coefficients are $a_{1,1}=-y$, $a_{1,2}=x$, $a_{1,3}=z$, $a_{1,4}=1$. The only non-vanishing term among the $A_{i,j}$ is $A_{1,2} = \frac{\partial a_{1,2}}{\partial q_1} - \frac{\partial a_{1,1}}{\partial q_2} = 1 - (-1) = 2$. All other components $A_{i,j}$ vanish identically. Following \eqref{muk1}, we compute the non-zero entries:
 \begin{align*}
\mu_{1,2} &= -A_{1,2}a_{1,3} + A_{1,3}a_{1,2} - A_{2,3}a_{1,1} = -(2)(z) + 0 - 0 = -2z, \\
\mu_{1,3} &= -A_{1,2}a_{1,4} + A_{1,4}a_{1,2} - A_{2,4}a_{1,1} = -(2)(1) + 0 - 0 = -2, \\
\mu_{2,3} &= -A_{1,3}a_{1,4} + A_{1,4}a_{1,3} - A_{3,4}a_{1,1} = 0 + 0 - 0 = 0,
 \end{align*}
yielding the following matrix representation:
 \begin{equation}
 \begin{pmatrix}\mu_{1,2} & \mu_{1,3} \\
 	& \mu_{2,3}\end{pmatrix} =\begin{pmatrix}-2z & -2 \\& 0\end{pmatrix}.
\end{equation}
 The presence of non-zero entries confirms that the system is strictly nonholonomic. These terms quantify the kinematic coupling between the rotation in the $(x, y)$ plane and the resulting drift in the $w$ and $z$ coordinates.
\end{exe}

\subsubsection*{The case $\kappa=2$}

\noindent
The mechanical system is subject to the pair of constraints
\begin{equation}
	\label{k2}
\left\{
\begin{array}{l}
a_{1,1}{\dot q}_1+a_{1,2}{\dot q}_2 + \dots + a_{1,n}{\dot q}_n=0
\\
a_{2,1}{\dot q}_1+a_{2,2}{\dot q}_2 + \dots + a_{2,n}{\dot q}_n=0
\end{array}
\right.
\end{equation}
For each $\nu=1,2$ the number (\ref{nc}) of independent conditions (\ref{fp}) is $\frac{(n-2)(n-3)}{2}$ and each term $\mu^{(\nu)}_{i,j}$, $i,j=1, \dots, n-2$, consists of six addends:
$$
\mu_{i,j}^{(\nu)}=A_{1,2+j}^{(\nu)}D_1(2+i)-A_{1,2+i}^{(\nu)}D_1(2+j)+
A_{2,2+j}^{(\nu)}D_2(2+i)-A_{2,2+i}^{(\nu)}D_2(2+j)-A^{(\nu)}_{2+i,2+j}D-A^{(\nu)}_{1,2}D_{1,2}(2+i,2+j)
$$
(observe that in the last summation of (\ref{mu}) only the values $r=1$ and $s=2$ are possible).

\noindent
To expand this formula, we denote the $2 \times 2$ determinants using subscripts that identify the corresponding columns of the constraint matrix:
$$
\begin{array}{l}
D_1(2+i)=det\,\left( \begin{array}{cc} a_{1, 2+i} & a_{1,2} \\ a_{2, 2+i} & a_{2,2} \end{array} \right)=D_{2+i, 2}, \quad 
D_2 (2+i)=det\,\left( \begin{array}{cc} a_{1,2} & a_{1,2+i} \\ a_{2, 1} & a_{2,2+i} \end{array} \right)=D_{1,2+i}, \\
\\ 
D_{1,2}(2+i, 2+j)=
det\,\left( \begin{array}{cc} a_{1, 2+i} & a_{1,2+j} \\ a_{2, 2+i} & a_{2,2+j} \end{array} \right)=D_{2+i, 2+j}, \quad D=D_{1,2}
\end{array}
$$
Substituting these into the general expression for $\kappa=2$, each coefficient $\mu_{i,j}^{(\nu)}$ (for $1 \leq i < j \leq n-2$) is given by:
\begin{equation}
	\label{mu_final_k2}
	\mu_{i,j}^{(\nu)} = A_{1,2+j}^{(\nu)}D_{2+i, 2} - A_{1,2+i}^{(\nu)}D_{2+j, 2} + A_{2,2+j}^{(\nu)}D_{1, 2+i} - A_{2,2+i}^{(\nu)}D_{1, 2+j} - A^{(\nu)}_{2+i,2+j}D_{1,2} - A^{(\nu)}_{1,2}D_{2+i,2+j}
\end{equation}

\noindent
For each $\nu=1,2$, the entries of the upper triangular matrix \eqref{matrmu} are structured as follows:

for $i=1$ ($j=2, \dots, n-2$):
\begin{equation}
	\label{muk2}
	\left\{
	\begin{aligned}
		\mu^{(\nu)}_{1,2} &=A^{(\nu)}_{1,4}D_{3,2}-A^{(\nu)}_{1,3}D_{4,2}
		+A^{(\nu)}_{2,4}D_{1,3}-A^{(\nu)}_{2,3}D_{1,4}
		-A^{(\nu)}_{3,4}D_{1,2}  -A^{(\nu)}_{1,2}D_{3,4} \\
		&\qquad \vdots \\
		\mu^{(\nu)}_{1,n-2} &=A^{(\nu)}_{1,n}D_{3,2}-A^{(\nu)}_{1,3}D_{n,2}
		+A^{(\nu)}_{2,n}D_{1,3} -A^{(\nu)}_{2,3}D_{1,n}
		-A^{(\nu)}_{3,n}D_{1,2}-A^{(\nu)}_{1,2}D_{3,n}
	\end{aligned}
\right.
\end{equation}
$\qquad \vdots \qquad \vdots \qquad \vdots$

\vspace{.5truecm}

for $i=n-4$ ($j=n-3, n-2$): 
\begin{equation*}
	\left\{
	\begin{aligned}
		\mu^{(\nu)}_{n-4,n-3} &= A^{(\nu)}_{1,n-1}D_{n-2,2}-A^{(\nu)}_{1,n-2}D_{n-1,2}
		+A^{(\nu)}_{2,n-1}D_{1,n-2}-A^{(\nu)}_{2,n-2}D_{1,n-1}
				-A^{(\nu)}_{n-2,n-1}D_{1,2}-A^{(\nu)}_{1,2}D_{n-2,n-1} \\
		\mu^{(\nu)}_{n-4,n-2} 
 &=A^{(\nu)}_{1,n}D_{n-2,2}-A^{(\nu)}_{1,n-2}D_{n,2}
 +A^{(\nu)}_{2,n}D_{1,n-2} -A^{(\nu)}_{2,n-2}D_{1,n} 
 -A^{(\nu)}_{n-2,n}D_{1,2}-A^{(\nu)}_{1,2}D_{n-2,n}
	\end{aligned}
\right.
\end{equation*}

for $i=n-3$ ($j=n-2$):
\begin{equation*} 
	\mu^{(\nu)}_{n-3,n-2} = A^{(\nu)}_{1,n}D_{n-1,2}-A^{(\nu)}_{1,n-1}D_{n,2}
		+A^{(\nu)}_{2,n}D_{1,n-1} -A^{(\nu)}_{2,n-1}D_{1,n} 
	-A^{(\nu)}_{n-1,n}D_{1,2}-A^{(\nu)}_{1,2}D_{n-1,n}
\end{equation*}

\begin{exe}
Consider the six-dimensional system of two points $P_1\equiv (x_1, y_1, z_1)$ and $P_2\equiv(x_2, y_2, z_2)$ whose ve\-lo\-ci\-ties are orthogonal to the straight line joining them: ${\dot P}_i\cdot (P_1-P_2)=0$, $i=1,2$.
Setting $(q_1, q_2, q_3, q_4 q_5, q_6)=(x_1, x_2, y_1, y_2, z_1, z_2)$, the constraints (\ref{k2}) are
($n=6, \kappa=2$)
$$
  \left\{
  \begin{array}{l}
  	(q_1-q_2){\dot q}_1 +  (q_3-q_4){\dot q}_3+(q_5-q_6){\dot q}_5=0
  	\\
  	(q_1-q_2){\dot q}_2+(q_3-q_4){\dot q}_4+(q_5-q_6){\dot q}_6=0
  \end{array}
  \right.
$$
yielding $a_{1,1}=a_{2,2}=q_1-q_2$, $a_{1,3}=a_{2,4}=q_3-q_4$, $a_{1,5}=a_{2,6}=q_5-q_6$, while all other $a_{\nu,i}=0$.
The relevant determinants are 
$$
\begin{array}{l}
D_{1,2}=(q_1-q_2)^2,\; \; D_{3,4}=(q_3-q_4)^2, \;\; D_{5,6}=(q_5-q_6)^2
\\
D_{1,4}=-D_{2,3}=(q_1-q_2)(q_3-q_4), \;\; D_{1,6}=-D_{2,5}=(q_1-q_2)(q_5-q_6), \;\;
D_{3,6}=-D_{4,5}=(q_3-q_4)(q_5-q_6)
\end{array}
$$
and all other determinants are zero. Condition (\ref{dik}) holds for $x_1\not = x_2$.
The only non-zero derivatives $A_{i,j}^{(\nu)}$ are $A_{1,2}^{(\nu)}=A_{3,4}^{(\nu)}=A_{5,6}^{(\nu)}=-1$ for both $\nu=1,2$. Applying \eqref{muk2}, we find 

$$
\mu_{i,j}^{(1)}=
\mu_{i,j}^{(2)}=
\left( 
\begin{array}{ccc}
	(q_1-q_2)^2+(q_3-q_4)^2	& 0 &(q_3-q_4)(q_5-q_6) \\
	& -(q_3-q_4)(q_5-q_6)& 0\\
	& & (q_1-q_2)^2+(q_5-q_6)^2
\end{array}
\right) \qquad  1\leq i<j \leq 4
$$
The matrix is non-zero unless the points coincide. This reflects the fact that while the distance $|P_1 - P_2|$ is conserved (a holonomic consequence), the orientation of the segment in space is subject to non-integrable constraints.
The configurations in which the maximum number of elements can be nullified are those where $y_1 = y_2$ or $z_1 = z_2$.
\end{exe}

\begin{exe}
Consider two points $P_1, P_2$ at a fixed distance $L$, where the midpoint velocity is parallel to the segment $P_1P_2$. Setting $(q_1, q_2, q_3, q_4)=(x_1, y_1, x_2, y_2)$, the constraints are:
	$$
	\left\{
	\begin{array}{l}
		(q_1-q_3){\dot q}_1+(q_2-q_4){\dot q}_2-(q_1-q_3){\dot q}_3-(q_2-q_4){\dot q}_4=0, \\
		-(q_2-q_4){\dot q}_1+(q_1-q_3){\dot q}_2-(q_2-q_4){\dot q}_3+(q_1-q_3){\dot q}_4=0
	\end{array}
	\right.
	$$
Condition (\ref{dik}) requires $D_{1,2}=(q_1-q_3)^2+(q_2-q_4)^2=\not =0$, which is violated if and only if the points coincide ($L=0$). 
The first constraint ($\nu=1$) is the derivative of the holonomic relation $(q_1-q_3)^2 + (q_2-q_4)^2 = L^2$, hence $\mu^{(1)}_{1,2}=0$. For the second constraint, $A_{1,2}^{(2)}=-2$ and $A_{3,4}^{(2)}=2$. Calculating \eqref{muk2}:
	\begin{equation*}
		\mu^{(2)}_{1,2} = -A^{(2)}_{3,4}D_{1,2} - A^{(2)}_{1,2}D_{3,4} = -4\left( (q_1-q_3)^2+(q_2-q_4)^2 \right) = -4L^2.
	\end{equation*}
	Whenever $L > 0$, $\mu^{(2)}_{1,2}$ is never zero, confirming the nonholonomic nature of the "parallel velocity" condition. As a side note, these two constraints imply the non-linear one $|{\dot P}_1|=|{\dot P}_2|$ \cite{talmeccanica}.
\end{exe}

\begin{exe}Consider a planar system where ${\dot P}_1$ is normal to $\overrightarrow{P_1P_2}$ and ${\dot P}_2$ is longitudinal to it. With $(q_1, q_2, q_3, q_4)=(x_1, x_2, y_1, y_2)$, we have:
$$
\left\{
\begin{array}{l}
	(q_1-q_2){\dot q}_1+(q_3-q_4){\dot q}_3=0, \\
	(q_3-q_4){\dot q}_2-(q_1-q_2){\dot q}_4=0
\end{array}
\right.
$$
and condition (\ref{dik}) is $D_{1,2}=(q_1-q_2)(q_3-q_4)\not =0$, that is we exclude those configurations where the points are aligned with the coordinate axes. In this example we again have $n=4$ and $k=2$. 
For $\nu=1$ the coefficients $A_{i,j}^{(1)}$ lead to $\mu^{(1)}_{1,2} = (q_1-q_2)^2+(q_3-q_4)^2$. For $\nu=2$, the specific structure of the derivatives leads to $\mu^{(2)}_{1,2} = 0$.
The resulting values are effectively the same as before; however, in this case, the constraint corresponding to $\nu=1$ is non-integrable. Even here, the system could be expressed through a single linear constraint, namely ${\dot P}_1 \cdot {\dot P}_2=0$. 
\end{exe}

\subsubsection*{The case $\kappa=n-1$}

\noindent
In view of (\ref{nc}), we expect formula (\ref{fp}) to yield no independent conditions, as $N_c=0$. We shall now verify this result. The only possible index values are $(i,j)=(1,1)$, which implies $\kappa+i=\kappa+j=n$. Consequently, (\ref{fp}) takes the form

\begin{equation*}
		\mu^{(\nu)}_{1,1} = \sum_{r=1}^{n-1} \left( A^{(\nu)}_{r, n} D_r(n) - A^{(\nu)}_{r, n} D_r(n) \right) + A^{(\nu)}_{n, n} D
		 - \sum_{1\leq r<s\leq n-1} A^{(\nu)}_{r,s} D_{r,s} (n,n)
\end{equation*}

\noindent
The expression on the right-hand side vanishes identically for all $\nu$. Specifically, the terms in the first sum cancel each other out, $A^{(\nu)}_{n,n} = 0$ by (\ref{ela}), and the determinants in the final sum are zero because they each contain two identical columns (both equal to $n$).
This case shows that when there are multiple constraints, the space of motion is so restricted that the condition is satisfied automatically.

\begin{exe}
Let us directly verify the validity of (\ref{derlagrsomma0}) in the simplest case $n=2$, $\kappa=1$, where only the single constraint $g_1=a_{1,1}{\dot q}_1+a_{1,2}{\dot q}_2=0$ is involved, with $a_{1,i}(q_1, q_2)$, $i=1,2$ being arbitrary functions. The corresponding condition (\ref{dcdeflin}) is $a_{1,1}\delta q_1+a_{1,2}\delta q_2=0$, therefore the two vectors $({\dot q}_1, {\dot q}_2)$ and $(\delta q_1, \delta q_2)$ have the same direction. On the other hand, the Lagrangian derivatives (\ref{derlagrlin}) take the form
$$
{\cal D}_1 g_1= 
\left( \dfrac{\partial a_{1,1}}{\partial q_2}-
\dfrac{\partial a_{1,2}}{\partial q_1}\right) {\dot q}_2, \quad 
{\cal D}_2 g_1= 
\left( \dfrac{\partial a_{1,2}}{\partial q_1}-
\dfrac{\partial a_{1,1}}{\partial q_2}\right) {\dot q}_1
$$
and clearly they satisfy ${\cal D}_1 g_1{\dot q}_1+{\cal D}_2 g_1{\dot q}_2=0$. Since, as already observed, $\delta {\bf q}$ is parallel to ${\dot {\bf q}}$, also ${\cal D}_1 g_1\delta q_1+{\cal D}_2 g_1\delta q_2=0$ holds, which is precisely (\ref{derlagrsomma0}).
\end{exe}

\section{A geometric interpretation of the conditions}

\noindent
Even though the conditions derived above may appear complex or purely analytical, they admit a concrete geometric interpretation. To see this, we begin with the following preliminary observations:
\begin{itemize}
\item Considering the generalized velocities as vectors in $\mathbb{R}^n$, the constraint conditions (\ref{glinom}) state that at any position ${\mathbf{q}}$, the velocity must be orthogonal to each row vector ${\mathbf{a}}_\nu$ ($\nu=1, \dots, \kappa$) of the matrix $a_{\nu,j}$:
\begin{equation}
\label{inf1}{\mathbf{a}}_\nu({\mathbf{q}})\cdot {\dot {\mathbf{q}}}=0, \qquad \nu=1, \dots, \kappa
\end{equation}

\item At any position ${\mathbf{q}}$, the vector $A^{(\nu)}({\mathbf{q}}){\dot {\mathbf{q}}}$ (where $A^{(\nu)}$ is defined in (\ref{ela})) is orthogonal to ${\dot {\mathbf{q}}}$ for any generalized velocity:
	\begin{equation}
		\label{inf2}
		{\dot {\mathbf{q}}}\cdot A^{(\nu)}({\mathbf{q}}){\dot {\mathbf{q}}}=0, \qquad \nu=1, \dots, \kappa
	\end{equation}
Indeed, the matrix $A^{(\nu)}$ is skew-symmetric; hence, its associated quadratic form vanishes identically.
	
\item The condition (\ref{zero2linvm}) is equivalent to requiring that the vector $A^{(\nu)}{\dot {\mathbf{q}}}$ lies within the subspace spanned by the vectors $\{{\mathbf{a}}_\nu\}$:
	\begin{equation}
		\label{inf3}
		A^{(\nu)}({\mathbf{q}}){\dot {\mathbf{q}}} \in \text{span} \{ {\mathbf{a}}_1, \dots, {\mathbf{a}}_\kappa \}
	\end{equation}
\end{itemize}

\begin{rem}
From a geometric perspective, the case $\kappa = n-1$ is a special one: the constraint equations (\ref{inf1}) define the velocity $\dot{\mathbf{q}}$ as a vector belonging to the orthogonal complement $\text{span} \{ {\mathbf{a}}_1, \dots, {\mathbf{a}}_{n-1} \}^\perp$. Since this orthogonal complement is one-dimensional, any vector orthogonal to $\dot{\mathbf{q}}$ must necessarily lie back in the $(n-1)$-dimensional subspace spanned by $\{{\mathbf{a}}_\nu\}$. Because $A^{(\nu)}\dot{\mathbf{q}}$ is always orthogonal to $\dot{\mathbf{q}}$ by (\ref{inf2}), the requirement (\ref{inf3}) is identically satisfied without further geometric restrictions on the field ${\mathbf{a}}_\nu$.
\end{rem}

\subsection{The case $n=3$}

\noindent
Consider the case $n=3$, which allows for an intuitive visualization in three-dimensional space. The left-hand side of equality (\ref{zero2linvm}) can be rewritten using the vector product: \begin{equation} 
\label{anabla} 
A^{(\nu)}{\dot {\mathbf{q}}}= 
\begin{pmatrix} 0 & A^{(\nu)}_{1,2} & A^{(\nu)}_{1,3} \\ -A^{(\nu)}_{1,2} & 0 & A^{(\nu)}_{2,3} \\ -A^{(\nu)}_{1,3} & - A^{(\nu)}_{2,3} & 0 
\end{pmatrix} 
\begin{pmatrix} {\dot q}_1 \\ {\dot q}_2 \\ {\dot q}_3 
\end{pmatrix}= 
(\nabla \wedge {\mathbf{a}}_\nu) \wedge {\dot {\mathbf{q}}}
\end{equation}
where $\wedge$ denotes the cross product in $\mathbb{R}^3$ and $\nabla \wedge {\mathbf{a}}_\nu$ is the curl of the vector field ${\mathbf{a}}_\nu$. It follows that $A^{(\nu)}{\dot {\mathbf{q}}}$ is orthogonal both to the velocity $\dot{\mathbf{q}}$ (consistently with (\ref{inf2})) and to the curl $\nabla \wedge {\mathbf{a}}_\nu$.

\begin{rem}
	If the constraints (\ref{glinom}) are integrable (i.e., they derive from a holonomic potential), then $\nabla \wedge {\mathbf{a}}_\nu = \mathbf{0}$ for all $\nu$. In this case, the expression in (\ref{anabla}) vanishes identically, and the condition (\ref{inf3}) is trivially satisfied. This confirms that the complex conditions we found are specifically characteristic of nonholonomic systems.
\end{rem}

\noindent
We distinguish two cases based on the number of constraints $\kappa$:
\begin{description}
	\item[Case $\kappa=1$:] From (\ref{inf1}), the velocity $\dot{\mathbf{q}}$ is orthogonal to ${\mathbf{a}}_1$. At the same time, by using (\ref{anabla}) we see that condition (\ref{zero2linvm}) is equivalent to
	\begin{equation}
		\label{nablak1}
		(\nabla \wedge {\mathbf{a}}_1) \wedge {\dot {\mathbf{q}}} = \varrho_1^{(1)}{\mathbf{a}}_1
	\end{equation}
The vector product (\ref{nablak1} can occur if and only if both vectors $\nabla \wedge {\mathbf{a}}_1$ and  ${\dot {\mathbf{q}}}$ lie in the plane perpendicular to vector ${\mathbf{a}}_1$, hence if
	\begin{equation}
		\label{a1wedge}{\mathbf{a}}_1 \cdot (\nabla \wedge {\mathbf{a}}_1) = 0
	\end{equation}
This coincides with the first condition in (\ref{muk1}).

	\item[Case $\kappa=2$:] The velocity $\dot{\mathbf{q}}$ is orthogonal to both ${\mathbf{a}}_1$ and ${\mathbf{a}}_2$, which are linearly independent. The subspace $\text{span}\{{\mathbf{a}}_1, {\mathbf{a}}_2\}$ contains all vectors orthogonal to $\dot{\mathbf{q}}$. Since $A^{(\nu)}\dot{\mathbf{q}}$ is always orthogonal to $\dot{\mathbf{q}}$ by (\ref{inf2}), it must necessarily lie in $\text{span}\{{\mathbf{a}}_1, {\mathbf{a}}_2\}$. 
	Thus, condition (\ref{inf3}) is automatically satisfied for any velocity, regardless of the direction of $\nabla \wedge {\mathbf{a}}_\nu$.
	The case $n=3, \kappa=2$ illustrates the general structural property whenever the system is governed by $\kappa = n-1$ constraints (see Remark 2): condition (\ref{fp}) holds unconditionally without requiring further assumptions.
\end{description}

\begin{exe}
A standard example in the literature is
\begin{equation}
	\label{noskid}
g_1={\dot q}_1 \sin q_3-{\dot q}_2 \cos q_3=0.
\end{equation}
This equation represents the no-skidding condition for a wheeled vehicle. It implies that the velocity of the contact point is always aligned with the body's orientation, meaning that lateral motion is prevented.	
The example falls under the case $n=3$, $k=1$, where
$$
{\bf a}_1=(\sin q_3,-\cos q_3,0), \quad A_{1,2}^{(1)}=0, \;\;A_{1,3}^{(1)}=\cos q_3, \;\; A_{2,3}^{(1)}=\sin q_3
\quad \nabla \wedge {\bf a}_1=(-\sin q_3, \cos q_3, 0)
$$ 
The non-singularity condition (\ref{dik}) corresponds to $\sin q_3 \not = 0$.
In this case, condition (\ref{a1wedge}) is written as
$$
(\sin q_3,-\cos q_3,0)\cdot (-\sin q_3, \cos q_3, 0)=-1
$$
hence the property (\ref{derlagrsomma0}) cannot hold.
If we add any constraint $g_2=a_{2,1}{\dot q}_1+a_{2,2}{\dot q}_2+a_{2,3}{\dot q}_3 =0$, the property becomes valid, as we enter the $n=3, k=2$ case. Let us verify this, so as to provide an explicit example: by means of (\ref{derlagrlin}) we determine
$$
{\cal D}_1 g_1 ={\dot q}_3\cos q_3, \;\;{\cal D}_2 g_1={\dot q}_3\sin q_3, \;\;{\cal D}_3g_1=-{\dot q}_1 \cos q_3 -{\dot q}_2 \sin q_3
$$
thus we need to check that the expression
\begin{equation}
\label{derlagrex6}
\sum\limits_{i=1}^3 {\cal D}_i g_1 \delta q_i=
{\dot q}_3 \cos q_3 \delta q_1 + {\dot q}_3 \sin q_3 \delta q_2 - \left(
{\dot q}_1 \cos q_3 + {\dot q}_2 \sin q_3 \right)\delta q_3
\end{equation}
vanishes for all $\delta q_1$, $\delta q_2$, $\delta q_3$ such that (see definition (\ref{derlagrsomma0lin}))
\begin{equation}
\label{deltaqex6}
\left\{
\begin{array}{l}
\sin q_3\delta q_1 - \cos q_3\delta q_2=0, 
\\
a_{2,1}\delta q_1 +a_{2,2}\delta q_2+a_{2,3} \delta q_3 =0.
\end{array}
\right.
\end{equation}
From the first of (\ref{deltaqex6}) and from (\ref{noskid}) we infer
$$
\delta q_1 = \frac{\cos q_3}{\sin q_3} \delta q_2, \quad {\dot q}_1 = \frac{\cos q_3}{\sin q_3}{\dot q}_2
$$
through which expression (\ref{derlagrex6}) becomes $\frac{1}{\sin q_3}({\dot q}_3 \delta q_2 - {\dot q}_2 \delta q_3)$. On the other hand, the constraint equation $g_2=0$ combined with the second condition in (\ref{deltaqex6}) yields $a_{2,3}({\dot q}_3 \delta q_2 - {\dot q}_2 \delta q_3)=-a_{2,1} ({\dot q}_1 \delta q_2- {\dot q}_2 \delta q_1)$ and the last expression is zero, as we immediately deduce from the first condition in (\ref{deltaqex6}) combined with (\ref{noskid}). One concludes that (\ref{derlagrex6}) vanishes for all virtual displacements that satisfy (\ref{deltaqex6}).
A similar argument shows that also $\sum\limits_{i=1}^3 {\cal D}_i g_2 \delta q_i=0$ for the set of displacements (\ref{deltaqex6}).
Regarding the second constraint, we can assume for instance that the wheel rolls without longitudinal slipping, which couples the forward velocity to the rotation. For the sake of simplicity, we consider the constraint relating the linear velocities to the angular rotation:
$$
g_2=\cos q_3 \dot{q}_1 + \sin q_3 \dot{q}_2 - r \dot{q}_3 = 0
$$
Here, the coefficients $a_{2,j}$ become $a_{2,1} = \cos q_3$, $a_{2,2} = \sin q_3$$a_{2,3} = -r$ (where $r$ denotes the radius). The coupling of constraint (\ref{noskid}) (no skidding) with $g_2=0$ (no slipping) satisfies property (\ref{derlagrsomma0lin}), as it belongs to the class of systems with $n=3$ and $k=2$.
\end{exe}

\begin{figure}[!h]
	\begin{center}
		\includegraphics[width=0.36\textwidth]{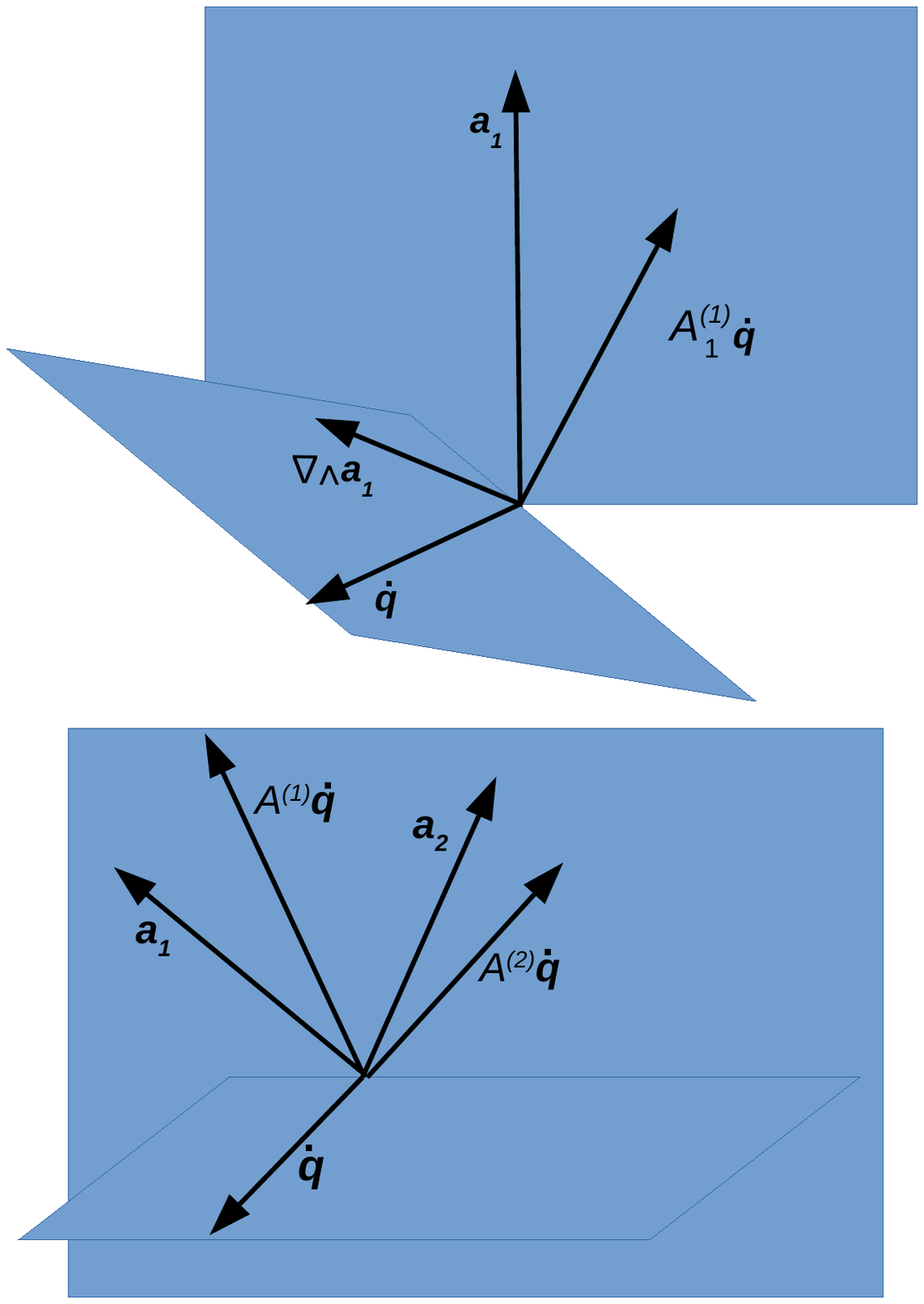}
		\caption{
upper image (case $n=3, \kappa=1$). Geometric interpretation for a single constraint in $\mathbb{R}^3$. The velocity $\dot{\mathbf{q}}$ is constrained to the plane orthogonal to $\mathbf{a}_1$. The vector $A^{(1)}\dot{\mathbf{q}}$, being the cross product of the curl $\nabla \wedge \mathbf{a}_1$ and $\dot{\mathbf{q}}$, is orthogonal to both. The condition (\ref{inf3}) requires $A^{(1)}\dot{\mathbf{q}}$ to be collinear with $\mathbf{a}_1$, which geometrically occurs if and only if the curl $\nabla \wedge \mathbf{a}_1$ is orthogonal to $\mathbf{a}_1$ itself. The picture illustrates that to ``align'' the constraint forces, the curl of the field must be perpendicular to the field itself. This is a geometric condition imposed on the system.
Lower image (case $n=3, \kappa=2$). Geometric interpretation for two constraints in $\mathbb{R}^3$. The velocity $\dot{\mathbf{q}}$ is restricted to the intersection of the planes orthogonal to $\mathbf{a}_1$ and $\mathbf{a}_2$ (a one-dimensional subspace). Any vector $A^{(\nu)}\dot{\mathbf{q}}$ orthogonal to $\dot{\mathbf{q}}$ (lying in the "vertical" plane in the sketch) is automatically a linear combination of $\mathbf{a}_1$ and $\mathbf{a}_2$. Thus, the condition is always satisfied regardless of the curl of the constraint fields. 
		}
		\label{fig:n3k12}
	\end{center}
\end{figure}

\begin{figure}[!h]
\begin{center}
\includegraphics[width=0.36\textwidth]{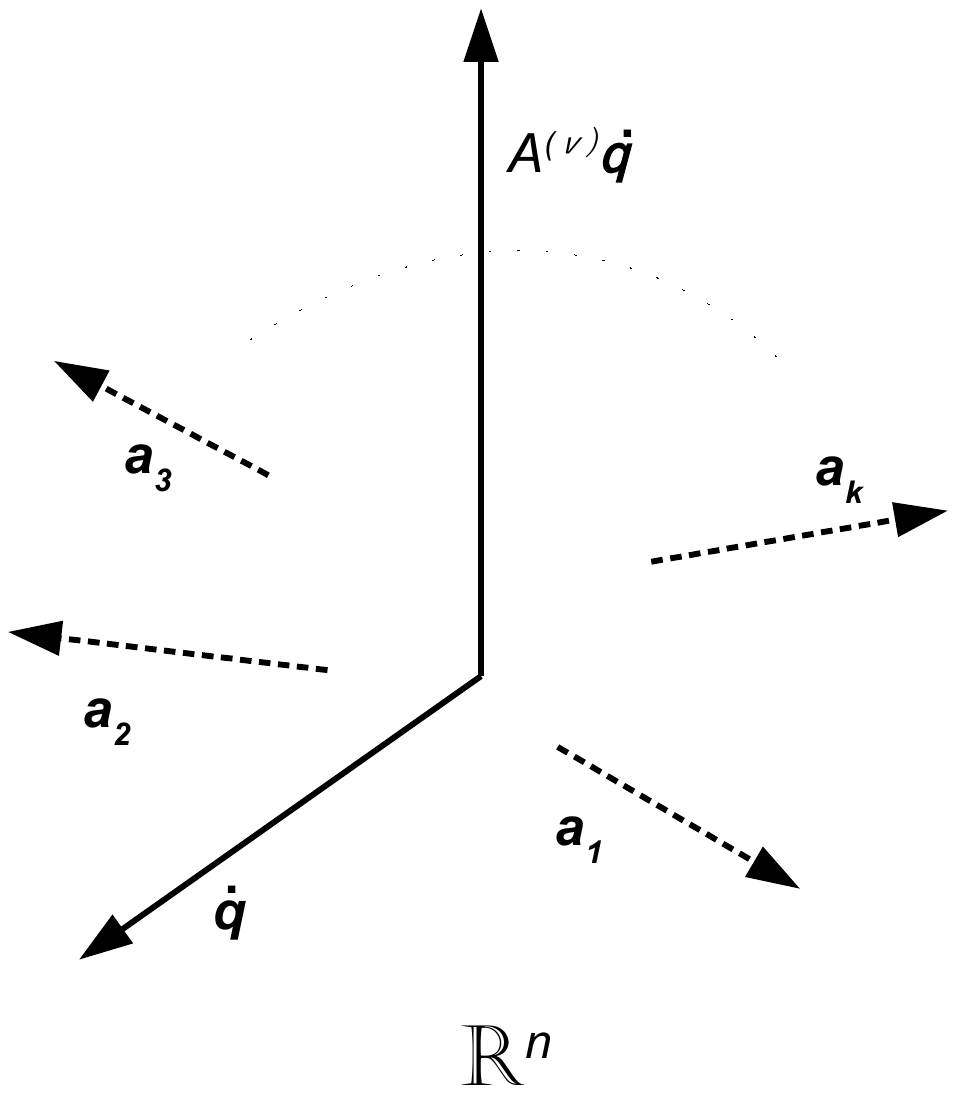}
\caption{
	 (general case $n > \kappa$): schematic representation in $\mathbb{R}^n$. The velocity $\dot{\mathbf{q}}$ is orthogonal to the $\kappa$-dimensional subspace $V = \text{span}\{\mathbf{a}_1, \dots, \mathbf{a}_\kappa\}$. The consistency of the dynamics requires that the vectors $A^{(\nu)}\dot{\mathbf{q}}$, which are inherently orthogonal to the velocity, "return" to lie within the subspace $V$.
The scheme aims to show how the vector $A\dot{\mathbf{q}}$ must ``fall back'' into the group of constraint vectors $\mathbf{a}_1 \dots \mathbf{a}_\kappa$.
}
\label{fig:nk}
\end{center}
\end{figure}

\section{Concluding remarks}

\subsection{Comparative analysis: condition (\ref{fp}) vs.~Frobenius integrability}

\noindent
It is instructive to compare the conditions derived in this work with the classical Frobenius Theorem. While both involve the curl (or Lie brackets) of the constraint vector fields, they address fundamentally different physical questions.The Frobenius Theorem determines whether a distribution of subspaces is integrable, which implies that the constraints are holonomic. For a single constraint ($\kappa=1$), the Frobenius condition is expressed as $\mathbf{a}_1 \cdot (\nabla \wedge \mathbf{a}_1) = 0$. As shown in Section 4, this is identical to our condition (\ref{fp}). In this specific case, the dynamic consistency of the system according to Chetaev’s principle is possible if and only if the constraint is integrable. However, a significant divergence emerges when $\kappa > 1$:
\begin{itemize}
\item \textbf{Frobenius Integrability:} requires the distribution spanned by $\{\mathbf{a}_1, \dots, \mathbf{a}_\kappa\}$ to be involutive. This imposes strict geometric restrictions on each vector field and their mutual Lie brackets, essentially requiring the distribution to be tangent to a family of nested manifolds.
\item \textbf{Condition (\ref{fp}):} does not require the distribution to be integrable. Instead, it allows for individual nonholonomic constraints, provided that their ``non-integrable parts'' interact in a specific way. These interactions are algebraically captured by the determinants of the coefficients in the overdetermined system.
\end{itemize}
As seen in the case $n=3, \kappa=2$, condition (\ref{fp}) is always satisfied, even if the constraints are strongly nonholonomic (i.e., they violate Frobenius). This suggests that Chetaev's dynamic consistency is a much broader requirement than geometric integrability. While Frobenius requires a ``static'' alignment of the constraint planes, condition (\ref{fp}) allows for a ``dynamic'' compensation between multiple nonholonomic constraints.

\subsubsection*{Dynamic compensation in nonholonomic systems}

\noindent
In conclusion, the comparison with Frobenius' Theorem highlights that the validity of condition (\ref{fp}) for $\kappa > 1$ should not be interpreted as a property of individual constraints, but rather as a phenomenon of collective coordination. While Frobenius integrability is a "static" property—requiring each constraint to individually or collectively define a manifold—condition (\ref{fp}) introduces what can be termed dynamic compensation between constraints.Individual vector fields may be nonholonomic and geometrically "disordered" when taken separately; however, if their interactions (captured by the determinants of the coefficients) are correctly balanced, the system maintains a global dynamic consistency according to Chetaev’s principle. This result significantly broadens the class of analyzable physical systems, suggesting that dynamic consistency is a much more resilient property than simple geometric integrability. In physical terms, this implies that in the presence of multiple constraints, the "deviations" from holonomy can cancel each other out, allowing the system to behave consistently even when the underlying geometry is intrinsically non-integrable.

\subsection{Conclusions}

\noindent
In this work, we have addressed a fundamental question in the analytical mechanics of nonholonomic systems: the characterization of constraint sets that satisfy the commutation property between the variation and the time derivative within the framework of Chetaev's postulate. 

The investigation led to the derivation of an explicit set of necessary conditions, algebraically encoded in the skew-symmetric matrix $\mu^{(\nu)}_{i,j}$. These conditions represent a bridge between the purely geometric properties of the constraint manifold (Frobenius integrability) and the dynamic consistency required by the d'Alembert-Lagrange principle. 

Our analysis highlights several key insights:
\begin{itemize}
	\item \textbf{Dynamic vs. Geometric Consistency}: For a single constraint ($\kappa=1$), the property \eqref{derlagrsomma0} coincides with the integrability of the constraint. However, for multiple constraints ($\kappa > 1$), the system may satisfy the Lagrangian derivative property even if the individual constraints are nonholonomic. This suggests that the interactions between constraints can "compensate" for their lack of integrability.
	
	\item \textbf{Structural Independence}: In systems with high constraints ($\kappa = n-1$), the property is satisfied regardless of the specific form of the constraints. This clarifies why many low-degree-of-freedom nonholonomic models found in the literature appear to satisfy variational identities that would otherwise be lost in higher dimensions.
	
	\item \textbf{Geometric Clarity}: Through the vector-based interpretation in $\mathbb{R}^3$, we have shown that the property \eqref{derlagrsomma0} is deeply linked to the alignment of the curl of the constraint fields with the subspace of the constraints themselves.
\end{itemize}

\subsection{Future perspectives}

\noindent
Future investigations will aim to extend this result in two main directions. First, we intend to explore the physical implications of the matrix $\mu$ in the context of the \textit{ideality} of constraints; specifically, understanding if these conditions identify a class of "quasi-holonomic" systems that admit simplified versions of the equations of motion. 

Second, the analysis could be generalized to non-linear constraints, where the presence of acceleration terms in the Lagrangian derivative introduces a dependency on the system's dynamics (the kinetic energy) rather than just the kinematics. This would open the way for a broader understanding of how nonholonomic systems can be categorized not only by their geometry but also by their variational behavior.

Finally, a significant contribution will be provided by the identification of systems that satisfy the progressively introduced hypotheses. This will serve to test whether the defined class of nonholonomic constraints is effectively broader than that of constraints with an integrating factor, and to examine the reciprocal effect of nonholonomic constraints through examples where the matrices $\mu^{\nu}_{i,j}$ vanish.

\appendix

\section{Proof of Theorem 1}

\noindent
The proof is structured as follows: we first establish the consistency of the overdetermined system using the method of bordered minors. We then perform a combinatorial expansion of the resulting determinants, finally utilizing a key Lemma to express the conditions in terms of independent generalized velocities.

\noindent
Conditions (\ref{zero2lin}) can be written as
\begin{equation}
	\label{betarho}
	{\bm \beta}^{(\nu)} = {\bm \varrho}^{(\nu)}, \qquad \nu=1, \dots, \kappa
\end{equation}
where ${\bm \beta}^{(\nu)}$ is the $n$-dimensional column vector ${\bm \beta}^{(\nu)} = A^{(\nu)} {\dot{\mathbf{q}}}$, with ${\dot{\mathbf{q}}} = ({\dot q}_1, \dots, {\dot q}_n)^T$,  $A^{(\nu)}$ the skew-symmetric matrix of order $n$ defined whose entries are defined in (\ref{ela}) and ${\bm \varrho}^{(\nu)}$ is the $n$-dimensional column vector 
$$
{\bm \varrho}^{(\nu)}=
\left(\sum\limits_{\mu=1}^\kappa \varrho_\mu^{(\nu)}a_{\mu, 1}, \cdots,
\sum\limits_{\mu=1}^\kappa \varrho_\mu^{(\nu)} a_{\mu, n}\right)^T.
$$

\noindent
We view (\ref{betarho}) as a linear system of $n$ equations in the $\kappa < n$ unknowns $(\varrho_1^{(\nu)}, \cdots, \varrho_\kappa^{(\nu)})$. By applying the basic principles of overdetermined systems, we know that a solution exists if and only if the rank of the $n \times (\kappa+1)$ augmented matrix $(A^{(\nu)}|{\bm \beta}^{(\nu)})$ remains $\kappa$, by virtue of assumption (\ref{dik}). To check this, it suffices to consider the $n-\kappa$ bordered minors of order $\kappa+1$ formed by the first $\kappa$ rows of the augmented matrix joined with the $(\kappa+i)$-th row for $i=1, \dots, n-\kappa$. These minors must vanish:
\begin{equation}
	\label{n-kdet}
	\det \begin{pmatrix} 
		a_{1,1} & \cdots & a_{1,\kappa} & \beta^{(\nu)}_1 \\
		\vdots & \ddots & \vdots & \vdots \\
		a_{\kappa,1} & \cdots & a_{\kappa,\kappa} & \beta^{(\nu)}_\kappa \\
		a_{\kappa+i,1} & \cdots & a_{\kappa+i,\kappa} & \beta^{(\nu)}_{\kappa+i}
	\end{pmatrix} = 0, \qquad i=1, \dots, n-\kappa
\end{equation}
For any set of $\kappa$ indices $\{i_1, i_2, \dots, i_\kappa\} \subset \{1, \dots, n\}$, not necessarily ordered, we define $D_{i_1, i_2, \dots, i_\kappa}$ as the determinant of the square submatrix formed by the corresponding columns ${\bf a}_{i_1}$, $\cdots$, ${\bf a}_{i_\kappa}$ of $(a_{\nu,j})$:
\begin{equation}
	\label{det_general}
	D_{i_1, i_2, \dots, i_\kappa} = \det \begin{pmatrix}
 {\bf a}_{i_1} & {\bf a}_{i_2} & \cdots & {\bf a}_{i_\kappa}
	\end{pmatrix}
\end{equation}

\noindent
This definition ensures that $D_{\sigma(i_1), \dots, \sigma(i_\kappa)} = \text{sgn}(\sigma) D_{i_1, \dots, i_\kappa}$ for any permutation $\sigma$. Consequently, $D$ is an alternating function of its indices, a property that will be extensively used to manipulate the sums in the following sections. In particular, it is worth noting that $D$ vanishes whenever any two indices are identical.

By applying the cofactor expansion along the $(\kappa+1)$-th column of the augmented matrix in (\ref{n-kdet}), we obtain:
\begin{equation}
	\label{condcsegno}
	\sum_{r=1}^{\kappa} (-1)^{\kappa+1+r} \beta^{(\nu)}_r D_{1, \dots, \hat{r}, \dots, \kappa, \kappa+i} + \beta^{(\nu)}_{\kappa+i} D_{1, \dots, \kappa} = 0, \quad i=1, \dots, n-\kappa
\end{equation}
where the hat notation $\hat{r}$ indicates that the index $r$ is omitted. 
To simplify (\ref{condcsegno}), we reorder the indices of the determinants by shifting the index $\kappa+i$ from the last position to the $r$-th position. This operation requires $\kappa-r$ transpositions, each contributing a factor of $-1$. Consequently, each term in the sum gains a total phase factor of:
\begin{equation*}
	(-1)^{\kappa+1+r} \cdot (-1)^{\kappa-r} = (-1)^{2\kappa+1-r+r} = (-1)^{2\kappa+1} = -1
\end{equation*}
This uniform sign change allows us to move the summation to the other side of the equality, resulting in the following equivalent form:
\begin{equation}
	\label{condc}
	\sum_{r=1}^{\kappa} \beta^{(\nu)}_r D_r (\kappa+i) = \beta^{(\nu)}_{\kappa+i} D, \quad i=1, \dots, n-\kappa
\end{equation}
where the determinants $D_r(\kappa+i)$ are defined in (\ref{drl}).

\begin{rem}
As an alternative approach, one may solve the first $\kappa$ equations in (\ref{betarho}) for the unknowns $\varrho^{(\nu)}_r$ using Cramer's Rule. Given the non-singularity assumption (\ref{dik}), we have 
$\varrho^{(\nu)}_r = \dfrac{D_r \left({\bm \beta}^{(\nu)}\right)}{D}$ where $D_r \left({\bm \beta}^{(\nu)}\right)$ is the determinant obtained by replacing  
the $r$-th column of the matrix associated with $D$ with the vector $(\beta^{(\nu)}_1, \dots, \beta^{(\nu)}_\kappa)^T$. Substituting these expressions into the remaining $n-\kappa$ equations (from position $\kappa+1$ to $n$) leads directly to (\ref{condc}), confirming the consistency of the Chetaev projection.
\end{rem}

\noindent
At this point, we express the conditions \eqref{condc} as polynomials in the generalized velocities $(\dot{q}_1, \dots, \dot{q}_n)$. 
By substituting in (\ref{condc}) the explicit expressions for the coefficients, $\beta^{(\nu)}_r = \sum_{s=1}^n A^{(\nu)}_{r,s} \dot{q}_s$, and exploiting the skew-symmetry of the matrix $A^{(\nu)}$, we can group the terms associated with each velocity $\dot{q}_s$:
\begin{equation}
\label{condcexpl}\sum_{s=1}^n \dot{q}_s \left( \sum\limits_{r=1}^\kappa A^{(\nu)}_{r,s} D_r(\kappa+i) - A^{(\nu)}_{\kappa+i,s} D \right) = 0.
\end{equation}
In expanded form, the coefficient for a generic velocity $\dot{q}_s$ in \eqref{condcexpl} becomes:
$$
\begin{array}{ll}
\textrm{for}\;s\leq  \kappa: & \sum\limits_{r=1}^\kappa A^{(\nu)}_{r,s} D_r(\kappa+i) + A^{(\nu)}_{s, \kappa+i} D\\
\\
\textrm{for}\; s = \kappa+j, \;j\in \{1, \dots, n-\kappa\}: &\sum\limits_{r=1}^\kappa A^{(\nu)}_{r,\kappa+j} D_r(\kappa+i) + A^{(\nu)}_{\kappa+j, \kappa+i} D
\end{array}
$$
Notice that the sign of the last term $A^{(\nu)}_{\kappa+j, \kappa+i}$ correctly accounts for the skew-symmetry $A^{(\nu)}_{\kappa+j, \kappa+i} = -A^{(\nu)}_{\kappa+i, \kappa+j}$ when the index $j$ exceeds $i$.

\noindent
We cannot yet conclude that the expressions in parentheses vanish identically, as the generalized velocities $\dot{q}_s$ are not independent due to the nonholonomic constraints \eqref{glinom}. These restrictions must be manipulated further to express the conditions solely in terms of a set of independent velocities.

\begin{lem}
Consider the system of $\kappa$ linear constraints given by $\sum\limits_{j=1}^n a_{\nu,j}\dot{q}_j = 0$, $\nu=1, \dots,\kappa$. 
For any choice of $\kappa-1$ distinct indices $\{i_1, \dots, i_{\kappa-1}\} \subset \{1, \dots, n\}$, the following identity holds:
\begin{equation}
\label{relv}
\sum_{h=1}^n \dot{q}_h D_{h, i_1, \dots, i_{\kappa-1}} = 0
\end{equation}where $D_{h, i_1, \dots, i_{\kappa-1}}$ denotes the determinant of the $\kappa \times \kappa$ matrix formed by the columns $\{h, i_1, \dots, i_{\kappa-1}\}$ of the constraint matrix.
\end{lem}

\noindent
This identity follows from the linearity of the constraints. Consider a $(\kappa \times \kappa)$ determinant where $\kappa-1$ columns are fixed. If we replace the remaining column with the vector result of the constraints, the determinant vanishes because the columns are linearly dependent for any admissible velocity. Expanding this null determinant via Laplace’s rule along the column of velocities recovers exactly the relation (\ref{relv}).

\noindent
The Lemma remains consistent even at the lowest constraint dimension $\kappa=1$: the set of indices $\{i_1, \dots, i_{\kappa-1}\}$ is empty. In this case, the determinants reduce to the individual coefficients, $D_h = a_{1,h}$, and the identity (\ref{relv}) becomes:$$\sum_{h=1}^n a_{1,h} \dot{q}_h = 0$$which is identical to the single constraint equation. This confirms that Lemma 1 is a natural generalization of the constraint equations themselves.

\noindent
Notice that the relation (\ref{relv}) is invariant (up to a sign) under any permutation of the indices $\{i_1, \dots, i_{\kappa-1}\}$. Thus, there are exactly $\binom{n}{\kappa-1}$ such relations.
Furtermore, any term where $h \in \{i_1, \dots, i_{\kappa-1}\}$ vanishes because the determinant would contain two identical columns. Consequently, an equivalent form is:
\begin{equation}
\label{relv2}\sum_{h=\kappa}^n \dot{q}_{i_h} D_{i_h, i_1, \dots, i_{\kappa-1}} = 0
\end{equation}
where $(i_1, \dots, i_n)$ is a permutation of $\{1, \dots, n\}$. 
\begin{exe}($n=4, \kappa=3$). There are $\binom{4}{2}=6$ relations. For the selection $\{i_1, i_2\} = \{3, 4\}$, we obtain:$$\dot{q}_1 D_{1,3,4} + \dot{q}_2 D_{2,3,4} = 0$$
This shows how the Lemma explicitly links the velocities of the independent coordinates to those of the dependent ones.
\end{exe}

\begin{exe} 
($n$ arbitrary, $\kappa=2$). The $n$ relations $\sum_{j=1}^n \dot{q}_j D_{j,i} = 0$ can be derived by taking the two constraint equations, multiplying the first by $a_{2,i}$ and the second by $a_{1,i}$. Their difference eliminates the $i$-th variable's coefficient, yielding the determinants $D_{j,i}$ and verifying the Lemma for the $\kappa=2$.
\end{exe}

\noindent
Lemma 1 serves as the fundamental tool for eliminating the dependent generalized velocities $\dot{q}_1, \dots, \dot{q}_\kappa$, allowing the conditions to be expressed solely in terms of the independent parameters $\dot{q}_{\kappa+1}, \dots, \dot{q}_n$. By applying this Lemma to the evolution equations, we ensure that the terms associated with the partial derivatives $\partial \phi_\nu / \partial q$ are not arbitrary, but represent the exact projection of the nonholonomic deviation onto the directions permitted by the constraints.

To complete the proof, we group the terms in (\ref{condcexpl}) related to $\dot{q}_1, \dots, \dot{q}_\kappa$ and $\dot{q}_{\kappa+i}$ (for a fixed $i \in \{1, \dots, n-\kappa\}$). Using the notation $D_r(\alpha)$ to denote the determinant where the $r$-th column is replaced by the $\alpha$-th column, the pairs are expressed as
\begin{equation}
\label{pair_dr}
\begin{cases}
A_{r,s} \left( \dot{q}_s D_r(\kappa+i) - \dot{q}_r D_s(\kappa+i) \right) & \text{for } 1 \leq r < s \leq \kappa \\
A_{p,\kappa+i} \left( \dot{q}_{\kappa+i} D_p(\kappa+i) + \dot{q}_p D \right) & \text{for } p = 1, \dots, \kappa
\end{cases}
\end{equation}
We now apply the identity (\ref{relv}) to the terms in parentheses. By exploiting the antisymmetry of the determinants and performing the necessary index permutations, the first group of pairs yields
\begin{equation}
\label{relvsr}\dot{q}_s D_r(\kappa+i) - \dot{q}_r D_s(\kappa+i) = (-1)^{\kappa+s-r} \left( \dot{q}_s D_{s, 1, \dots, \hat{r}, \dots, \hat{s}, \dots, \kappa, \kappa+i} + \dot{q}_r D_{r, 1, \dots, \hat{r}, \dots, \hat{s}, \dots, \kappa, \kappa+i} \right)
\end{equation}
where the notation $\hat{\cdot}$ indicates the omission of the corresponding index. Applying (\ref{relv2}) with the fixed set of $\kappa-1$ indices $(1, \dots, \hat{r}, \dots, \hat{s}, \dots, \kappa, \kappa+i)$, the expression is replaced by a sum over the independent velocities:
\begin{equation}
= - (-1)^{\kappa+s-r} \sum\limits_{j =1, \dots, n-\kappa, j \neq i} \dot{q}{\kappa+j} D_{\kappa+j, 1, \dots, \hat{r}, \dots, \hat{s}, \dots, \kappa, \kappa+i}
\end{equation}
Similarly, for the second group ($p=1, \dots, \kappa$), the relation becomes
\begin{equation}
\label{relvp}
\dot{q}_{\kappa+i} D_p(\kappa+i) + \dot{q}_p D = (-1)^{p} \sum\limits_{j =1, \dots, n-\kappa, j \neq i} \dot{q}_{\kappa+j} D_{\kappa+j, 1, \dots, \hat{p}, \dots, \kappa}
\end{equation}
The relations (\ref{relvsr}) and (\ref{relvp}) demonstrate that all pairs in (\ref{pair_dr}) can be expressed as linear combinations of the independent velocities $\dot{q}_{\kappa+1}, \dots, \dot{q}_n$.Substituting these results back into the full expression (\ref{condcexpl}) and collecting the coefficients for each $\dot{q}_{\kappa+j}$, we obtain
\begin{equation}
\sum\limits_{j=1}^{n-\kappa} \mu^{(\nu)}_{i,j} \dot{q}_{\kappa+j} = 0 \quad \text{for } i=1, \dots, n-\kappa, \quad \nu=1, \dots, \kappa
\end{equation}
where $\mu^{(\nu)}_{i,j}$ are exactly the coefficients defined in (\ref{mu}). Since the velocities $\{\dot{q}_{\kappa+j}\}$ are independent, their coefficients must vanish identically, confirming that the conditions (\ref{fp}) hold. $\quad \square$


\end{document}